\documentclass[reqno,12pt,letterpaper]{amsart}
\usepackage{amsmath,amssymb,amsthm,graphicx,mathrsfs,url}
\usepackage[usenames,dvipsnames]{color}
\usepackage[colorlinks=true,linkcolor=Red,citecolor=Green]{hyperref}
\usepackage{amsxtra}
\usepackage{tikz}
\usepackage{amscd}
\usetikzlibrary{arrows.meta}
\usepackage[normalem]{ulem}
\usepackage{wasysym} 
\usepackage{graphicx}
\usepackage{subcaption}
\usepackage[normalem]{ulem}
\def\arXiv#1{\href{http://arxiv.org/abs/#1}{arXiv:#1}}
\usepackage{array}
\newcolumntype{P}[1]{>{\centering\arraybackslash}m{#1}}

\def\?[#1]{\textbf{[#1]}\marginpar{\Large{\textbf{??}}}}

\let\epsilon=\varepsilon 
\newcommand{\CC}{{\mathbb C}}

\newtheorem{prop}{Proposition}

\newtheorem{lemm}[prop]{Lemma}

\numberwithin{equation}{section}

\DeclareMathOperator{\Spec}{Spec}

\let\Im=\Imag

\let\Re=\Real

\usepackage{scalerel}

\newcommand\reallywidehat[1]{\arraycolsep=0pt\relax%
\begin{array}{c}
\stretchto{
  \scaleto{
    \scalerel*[\widthof{\ensuremath{#1}}]{\kern-.5pt\bigwedge\kern-.5pt}
    {\rule[-\textheight/2]{1ex}{\textheight}} 
  }{\textheight} %
}{0.5ex}\\           
#1\\                 
\rule{-1ex}{0ex}
\end{array}
}

\begin{document}


\title[From the chiral to Bistritzer--MacDonald models]{
From the chiral model of TBG to the Bistritzer--MacDonald model}

\author{Simon Becker}
\email{simon.becker@math.ethz.ch}
\address{ETH Zurich, 
Institute for Mathematical Research, 
Rämistrasse 101, 8092 Zurich, 
Switzerland}

\author{Maciej Zworski}
\email{zworski@math.berkeley.edu}
\address{Department of Mathematics, University of California,
Berkeley, CA 94720, USA}

\begin{abstract}
We analyse the splitting of exact flat bands in the chiral model of the
twisted bilayer graphene (TBG) when the $AA'/BB'$ coupling of the full Bistritzer--MacDonald 
model is taken into account. The first-order perturbation caused by the $AA'/BB'$ 
potential the same for both bands and satisfies interesting 
symmetries (see \eqref{eq:prope}),  in particular it vanishes on the line defined by the $K$ points. The splitting of the flat bands is governed by the quadratic term which vanishes at the $K$ points. 
\end{abstract}

\maketitle

\section{Introduction and statement of the main result}
\label{s:int}

The Bistritzer--MacDonald Hamiltonian (BMH)  \cite{BM11} is considered a good model for
the study of twisted bilayer graphene (TBG) and is famous for an accurate prediction 
of the twisting angle at which superconductivity occurs \cite{Cao} (see \cite{CGG, Wa22} for its rigorous derivation). \begin{figure}
\includegraphics[width=16cm]{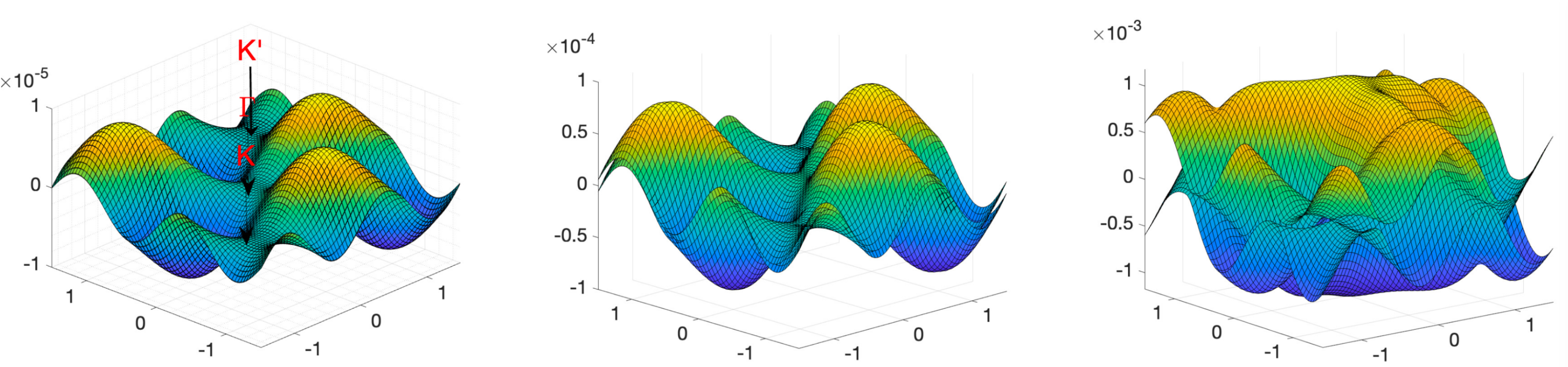}
\caption{\label{f:2} Plots of $ k \mapsto E_{\pm 1} ( \alpha, \lambda, k ) $ for 
$ \alpha $ the first real magic element of $ \mathcal A $ and 
$ \lambda =  10^{-3}, 10^{-2}, 10^{-1 } $. We see that for very small coupling the flat bands
``move together" and split only when the coupling gets larger; the quadratic term controls
the splitting of the bands, see Figure~\ref{f:1}. For animated versions see 
\url{https://math.berkeley.edu/~zworski/Chiral2BM.mp4} and
\url{https://math.berkeley.edu/~zworski/1band_1D.mp4}.
}
\end{figure}
In our notation it takes
the form
\begin{equation}
\label{eq:defBM} H ( \alpha, \lambda ) := \begin{pmatrix} \lambda C & D ( \alpha)^* \\
D ( \alpha ) & \lambda C \end{pmatrix} : H^1 ( \mathbb C ; \mathbb C^4 ) \to 
L^2 ( \mathbb C ; \mathbb C^4 ),  \ \ \alpha \in \mathbb C, \ \ \lambda \in \mathbb R,  \end{equation}
where $ D ( \alpha ) $ and $ C $ are defined in \eqref{eq:Hamiltonian}. The bands
for \eqref{eq:defBM} are given as the eigenvalues of $ e^{- i \langle k , z \rangle } 
H ( \alpha, \lambda ) e^{ i \langle k , z \rangle } $ acting on $ L^2_0 $ (a subspace 
of $ L^2_{\rm{loc}} ( \mathbb C ; \mathbb C^4 ) $ satisfying a periodicity condition -- see \eqref{eq:Lk}). 
The chiral limit of BMH corresponds to $ \lambda = 0$ and its elegant mathematical 
properties have been exploited first by Tarnopolsky--Kruchkov--Vishwanath \cite{magic}, then mathematically by Watson--Luskin \cite{lawa}, and the authors and their collaborators \cite{phys, beta, bhz1, bhz2, bhz3, yang, gaz, BHWY, hiz}. 

The bands in the chiral model, that is the eigenvalues of  $ e^{- i \langle k , z \rangle } 
H ( \alpha, 0 ) e^{ i \langle k , z \rangle } $ on $ L^2_0 $,  are symmetric with respect to 
$ 0 $ and we can label them as $ E_{\pm \ell } ( \alpha, 0 , k ) $, $ E_{-\ell } = - E_{\ell} $, 
$ E_1 \geq 0 $. 
$ \ell \in \mathbb Z \setminus \{ 0 \} $. Following \cite{beta, bhz2}, there exists 
a discrete set of {\em magic} $ \alpha$'s, $ \mathcal A \subset \mathbb C $, such that $ E_{\pm 1} ( \alpha , 0 , k ) \equiv 0 $  for all $k \in \CC$, 
if and only if $ \alpha \in \mathcal A $ (the angle of twisting is proportional to $ 1/\alpha $). A magic $ \alpha $ is {\em simple} if $ E_{\ell } ( \alpha, 0 , k ) 
> 0$ for $ \ell > 1 $ -- see \cite{bhz2} and, for the existence and properties of degenerate
magic $ \alpha $'s, \cite{bhz3}.

Since flat bands in the chiral limit are uniformly gapped from the rest of the spectrum \cite{bhz2,bhz3}, we can still consider, for small $\lambda$, the bands 
$ E_{\pm \ell} ( \alpha, \lambda, k ) $, where for $ \alpha \in \mathcal A $, $ E_{\pm \ell } ( \alpha, 0 , k ) \equiv 0 $. 
The purpose of this note is to use symmetries of BHM \eqref{eq:defBM}, some
properties of theta functions, 
and standard degenerate perturbation theory to show 

\noindent
{\bf Theorem.} {\em Suppose that $ \alpha \in \mathcal A \cap \mathbb R $ is simple and 
that $ k \mapsto E_{\pm 1 } ( \alpha, \lambda , k ) $ are the two lowest bands (in 
absolute value) of BMH in 
\eqref{eq:defBM}. Then there
exist  $e(\alpha,\bullet),f(\alpha,\bullet)\in C^{\infty}(\mathbb C/\Lambda^*)$ such that
\begin{equation}
\label{eq:Epm1}  E_{\pm 1 } ( \alpha, \lambda , k ) = e ( \alpha, k ) \lambda 
\pm |f ( \alpha, k )| \lambda^2 + 
 \mathcal O ( \lambda^3 ) , \ \ 
\lambda \to 0 , \end{equation}
$ f ( \pm K ) = 0 $, ($ \omega K \equiv K \!\! \! \mod \! \Lambda^*$, $ K \neq 0 $), and 
\begin{equation}
\label{eq:prope}   e (\alpha,   k ) = -e ( \alpha,  - k ) = - e( \alpha, \bar k ) = e(\alpha,  \omega k ) , \ \ \omega = e^{ { 2 \pi i } /3} .  
\end{equation}}

\noindent
{\bf Remarks.} 1. A  surprising feature, which does not seem to have been observed
before, is that the linear term does not depend on the band: the two bands move together for small 
coupling constants -- see Figure~\ref{f:2}.

\noindent
2. We observe numerically that the coefficient of the linear term is much smaller than 
the coefficient of the quadratic term and that the approximation \eqref{eq:Epm1} is
reasonable for physically relevant coupling $ \lambda \simeq 0.7\alpha$, for
first real magic $ \alpha$ -- see Figure~\ref{f:1}. 

\noindent
3. The situation is different for complex $ \alpha$'s showing that our assumption 
$ \alpha \in \mathcal A \cap \mathbb R$ is necessary to have purely quadratic
behaviour in band splitting -- see Figure~\ref{f:3}. 

\noindent
4. Similar analysis can be performed for double magic $ \alpha $ but to keep things
simple we restrict ourselves to numerical illustrations -- see \S \ref{s:double}.

We now review the definitions needed for the notation of \eqref{eq:defBM} and the theorem.

We recall that general version of $ D ( \alpha ) $ (and the corresponding self adjoint 
Hamiltonian) and of $ C $: 
\begin{equation}
\label{eq:Hamiltonian}
 D(\alpha) = \begin{pmatrix} 2D_{\bar z }& \alpha U (z ) \\ \alpha U(- z) & 2 D_{\bar z } \end{pmatrix}, \ \ \ \
  C := \begin{pmatrix} 0 & V ( z )  \\
V  ( -z ) & 0 \end{pmatrix} , 
\end{equation}
where the parameter $\alpha$ is proportional to the inverse relative twisting angle. 
With $\omega = e^{2\pi i /3}$ and  $ K := \frac 4 3 \pi $,  we assume that 
\begin{equation}
\label{eq:defU}
U(z  + \gamma ) = e^{ i \langle \gamma, K \rangle } U  (z ), \ 
\gamma \in \Lambda ,  \quad U  (\omega z ) = \omega U  (z ), \ \ 
 \overline{U( \bar z ) } = - U ( - z ), \end{equation}
$ \Lambda := \mathbb Z \oplus \omega \mathbb Z $, and 
 \begin{equation}
 \label{eq:defV}  V (  z ) = V ( \bar z ) = \overline{ V ( - z ) } ,  \ \ V ( \omega z ) = V ( z ) , 
 \ \ V ( z + \gamma ) = e^{ i \langle \gamma, K \rangle } V ( z) . 
 \end{equation}

The specific potentials (see \cite{phys} and references given there) and the ones 
we use in numerical experiments are 
\begin{equation}
\label{eq:defUV}  U ( z ) =  - \tfrac{4} 3 \pi i \sum_{ \ell = 0 }^2 \omega^\ell e^{ i \langle z , \omega^\ell K \rangle },
\ \ V( z) = \sum_{ \ell = 0 }^2  e^{ i \langle z , \omega^\ell K \rangle }, \ \
 K = \tfrac43 \pi  .
 \end{equation}

\begin{figure}
\includegraphics[width=7.3cm]{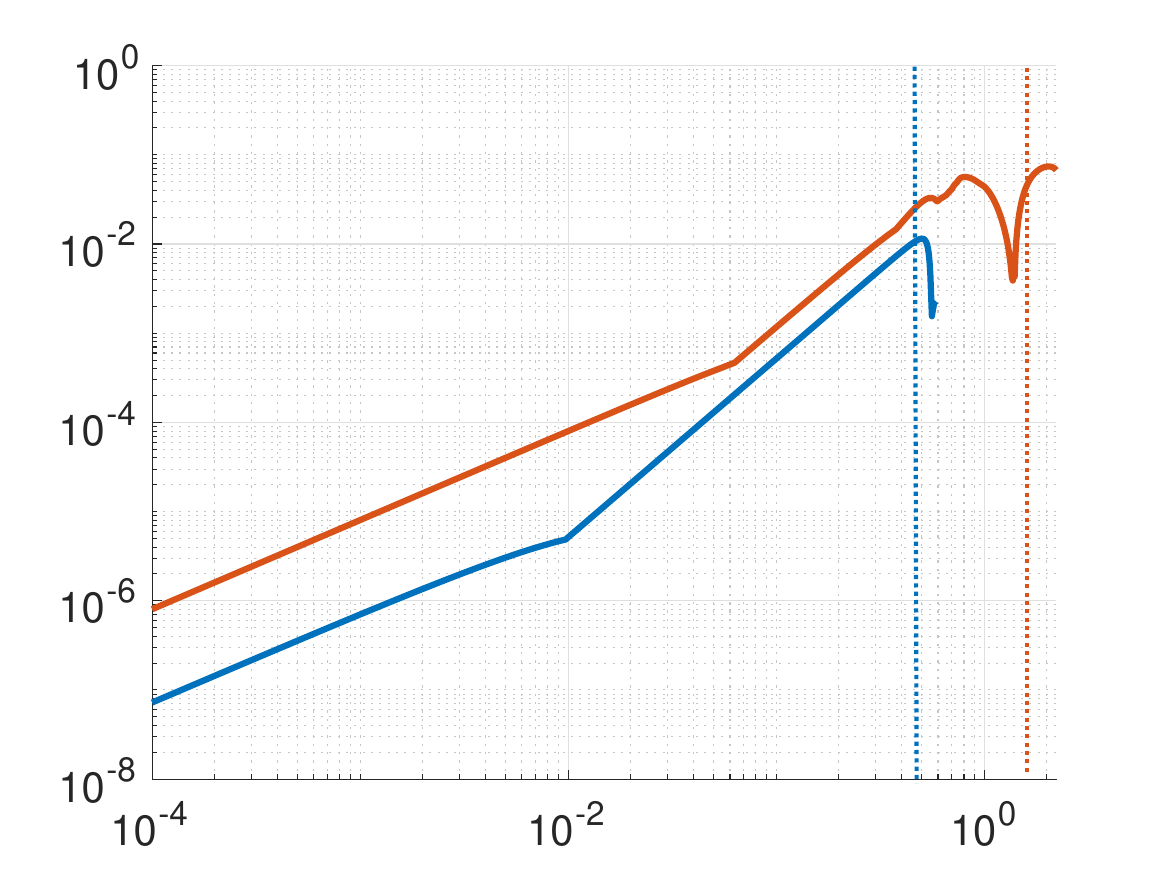} \hspace{0.25cm}
\includegraphics[width=7.3cm]{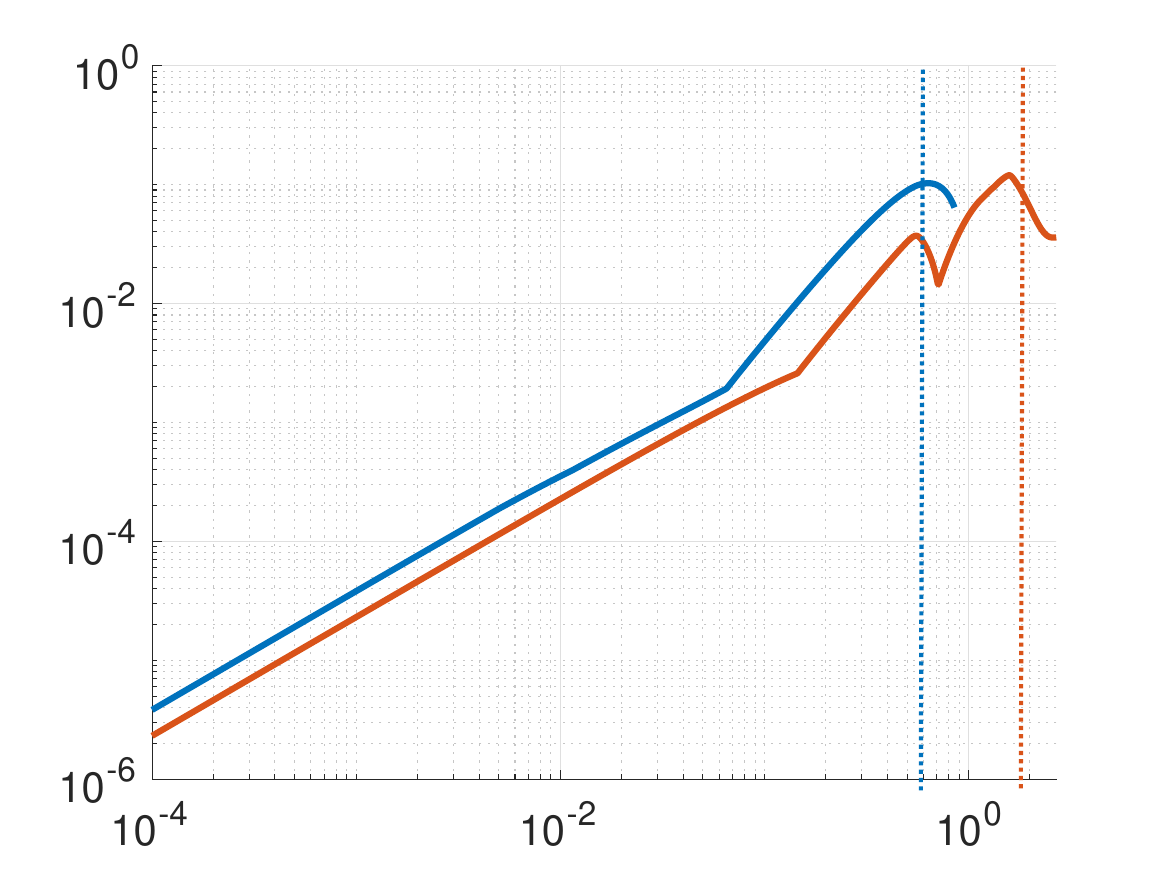}\hspace{0.25cm} 
\caption{\label{f:1} On the left a log-log plot of $ \lambda \mapsto 
\max_k | E_1 ( \alpha, \lambda , k ) | $ which shows the transition from the 
linear to quadratic dominant behaviour for the 
first two magic  $ \alpha \simeq \{ {\color{blue}0.586}, \,  {\color{red}2.221}\};$
on the right a log-log plot of $ \lambda \mapsto 
| E_1 ( \alpha, \lambda , 0 ) |$ for two {\em degenerate} magic angles $ \alpha \simeq \{ {\color{blue}0.8538}, \,  {\color{red}2.701}\}$. The choice $\lambda \approx 0.7 \alpha$ is indicated by vertical dashed lines. The potentials used on the left are given in
\eqref{eq:defUV} and correspond to the standard BMH. 
On the right we use $ U_1 ( z ) :=   ( U ( z) - U ( 2z ) )/\sqrt 2 $ and $ V ( z ) $. 
That potential exhibits double magic $ \alpha$'s on the real axis -- see \cite{bhz3}.}
\end{figure}

\begin{figure}
\includegraphics[width=7.3cm]{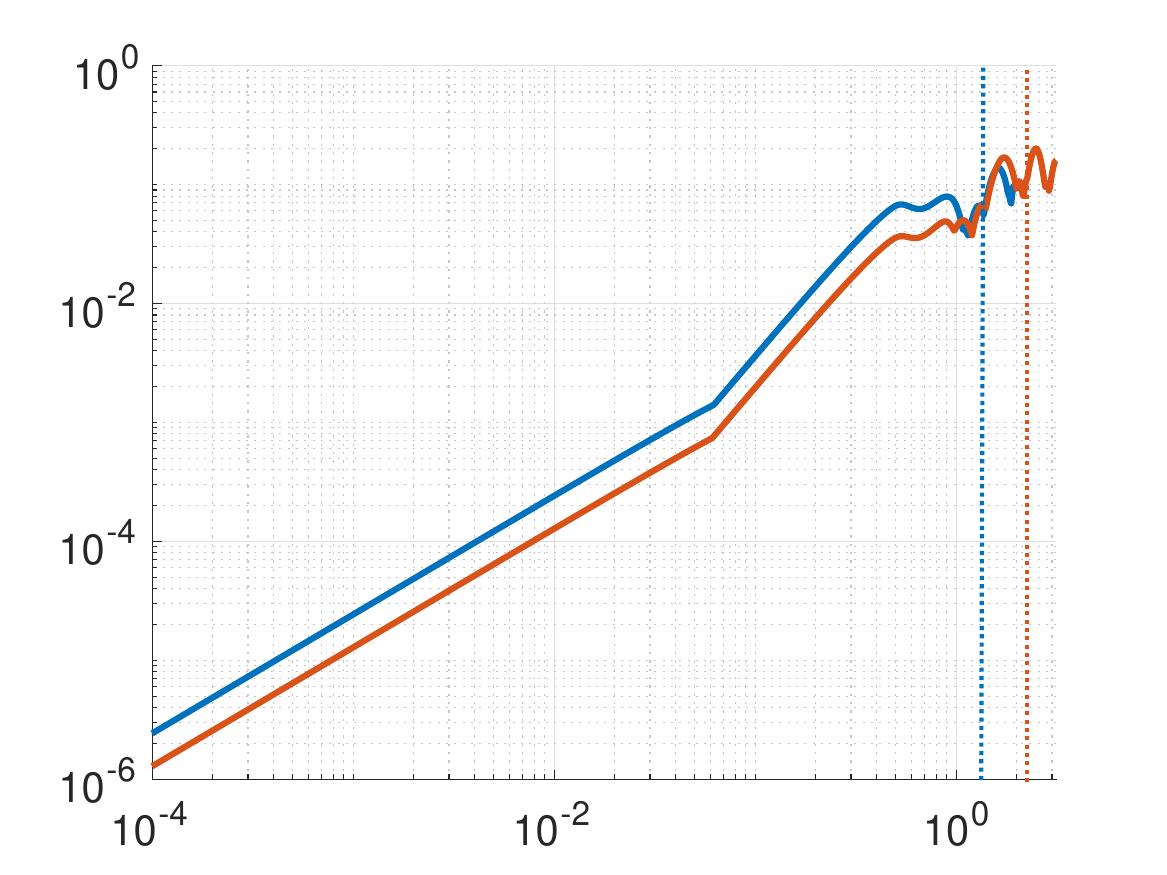} \hspace{0.25cm}
\includegraphics[width=7.3cm]{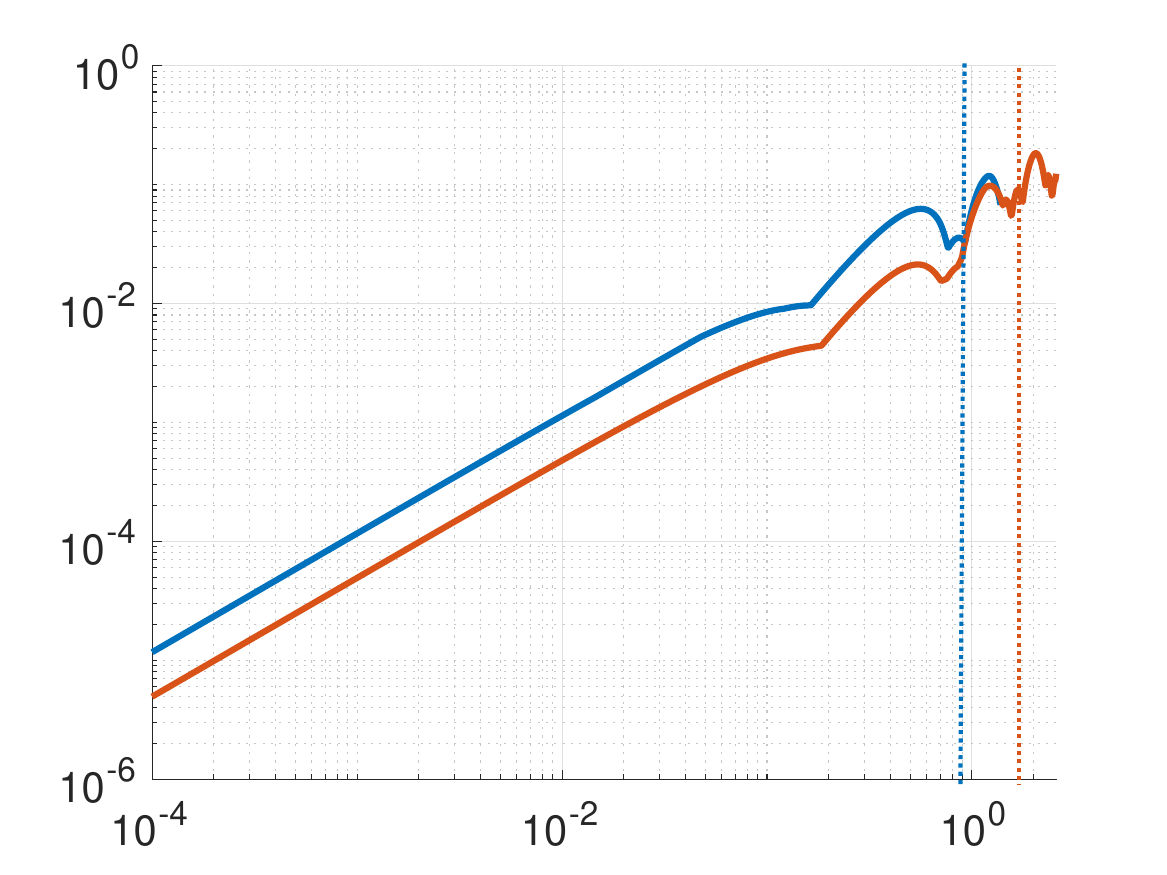}
\caption{\label{f:3} Left: log-log plot of $\lambda \mapsto 
\max_k | E_1 ( \alpha, \lambda , k ) | $ 
which shows the simple complex magic angles $ \alpha \simeq \{{\color{blue}1.121 + 1.57i}, \, {\color{red}1.312 + 2.862i}\}$;
Right: log-log plot of $ \lambda \mapsto \max_k
| E_1 ( \alpha, \lambda , k ) |$ for two {\em two-fold degenerate} complex magic
$ \alpha \simeq \{{\color{blue}0.963 + 0.987i}, \, {\color{red}1.226 + 2.268i}\} $ both with $U_1$ and $V$ as in Figure \ref{f:1}. Vertical lines at $ |\lambda/\alpha| = 0.7$}
 
\end{figure}

We conclude this introduction by describing the relation between $ H ( \lambda, \alpha ) $
given in \eqref{eq:defBM} and the representation of that Hamiltonian in the physics 
literature \cite{BM11}, \cite{magic}. 
The Bistritzer-MacDonald Hamiltonian for twisted bilayer graphene reads 
 \begin{equation}
 \label{eq:mathscrH}
 \mathscr H(\beta,\zeta)  = \begin{pmatrix} -i\boldsymbol{\sigma}\cdot\nabla & T(\beta,\zeta)  \\  
T(\beta,\zeta) ^\dag & -i\boldsymbol{\sigma}\cdot \nabla \end{pmatrix},
\end{equation}
where we defined $\boldsymbol{\sigma}= (\sigma_x,\sigma_y)$, $\mathbf r = (x,y)$, 
$ \boldsymbol{\sigma}\cdot\nabla = \sigma_x \partial_x + \sigma_y \partial_y $ ($ \sigma_\bullet $
are the Pauli matrices) 
 and interlayer tunnelling matrix\[T(\beta,\zeta \mathbf r) = \begin{pmatrix}\tilde \beta V(\zeta \mathbf r) &\beta \overline{U(-\zeta \mathbf r)} \\  \beta U(\zeta \mathbf r) & \tilde \beta V(\zeta \mathbf r) \end{pmatrix}.\]
The honeycomb graphene lattice has a non-equivalent pair of atoms in its fundamental domain labelled by $A,B$ respectively with $A',B'$ used for the second layer in case of TBG.
In the matrix $T$, the potentials $U$ and $V$ facilitating $AB'$ tunnelling and $AA'/BB'$ respectively, see Figure \ref{fig:tunnelling}, are given for $q_i = R^i (0,-1)$ with $2\pi/3$ rotation matrix $R = \frac{1}{2}\begin{pmatrix} -1 & -\sqrt{3} \\ \sqrt{3} & -1 \end{pmatrix}$ by
\[ 
\begin{gathered} U(\mathbf r) = \sum_{i=0}^2 \omega^i e^{-i q_i\cdot \mathbf r}, \ \ 
V(\mathbf r) = \sum_{i=0}^2  e^{-i q_i\cdot \mathbf r}, \ \
 q_i : = R^i  \begin{pmatrix} \ \ 0 \\ -1\end{pmatrix} , \ \  R := \tfrac{1}{2}\begin{pmatrix} -1 & -\sqrt{3} \\ \sqrt{3} & -1 \end{pmatrix}. 
\end{gathered} 
\]
($ R $ is the $ 2 \pi /3 $ rotation matrix.)

Let $K = \operatorname{diag}(1,\sigma_x,1)$ then upon making the change of variables $\zeta \mathbf r \mapsto \mathbf r $, $\lambda = \beta /\zeta$, and $\alpha = \beta /\zeta$, we find
\[ K  \mathscr H(\beta,\zeta) K = \zeta \begin{pmatrix} \lambda C & D(\alpha)^{\dagger} \\  D(\alpha) & \lambda C \end{pmatrix}\]
where 
\[\begin{split} 
D(\alpha) &= \begin{pmatrix}D_{x} + i D_{y} & \alpha U(\mathbf r) \\ \alpha U(-\mathbf r) & D_{x}+i D_{y}\end{pmatrix} \text{ and }
C =\begin{pmatrix}0 & V(\mathbf r) \\ V(-\mathbf r) & 0 \end{pmatrix}. 
\end{split}\] 
The physics coordinates $\mathbf r = (x,y)$ are related to our Hamiltonian \eqref{eq:defBM} by the change of coordinates
\[ x + i y=\tfrac{4}{3}\pi i z, \quad \tfrac{1}{3}\Gamma=\tfrac{4}{3}\pi i \Lambda, \quad 3\Gamma^*=\frac{3}{4\pi i }\Lambda^*\]
where $\tfrac{1}{3}\Gamma$ is the moir\'e lattice and $ \Gamma $ is the lattice of periodicity of 
the potentials $U ( \mathbf r ) $ and $ V ( \mathbf r ) $. (In the figures we use coordinates based on $(x,y)$
and the corresponding $ k$ coordinates.)

\noindent
{\bf Notation.} In this paper we use the physics notation: for an
operator $ A $ on $ L^2 ( M, dm )  $, $ \langle u | A | v \rangle := \int_M Av \, \bar u \, dm $.
Also, $ | u \rangle $ denotes the operator $ \mathbb C \ni \mu \to \mu u \in L^2 $ and
$ \langle u | $, its adjoint $ L^2 \ni v \to \langle u | v \rangle \in \mathbb C$. 
For $ z, w \in \mathbb C \simeq \mathbb R^2 $, we use 
the real inner product, $ \langle z , w \rangle := \Re z \bar w$. 

\noindent
{\bf Acknowledgements.} MZ gratefully acknowledges partial support by the NSF grant DMS-1901462
and by the Simons Foundation under a ``Moir\'e Materials Magic" grant. 

\begin{figure}
\includegraphics[width=5.3cm]{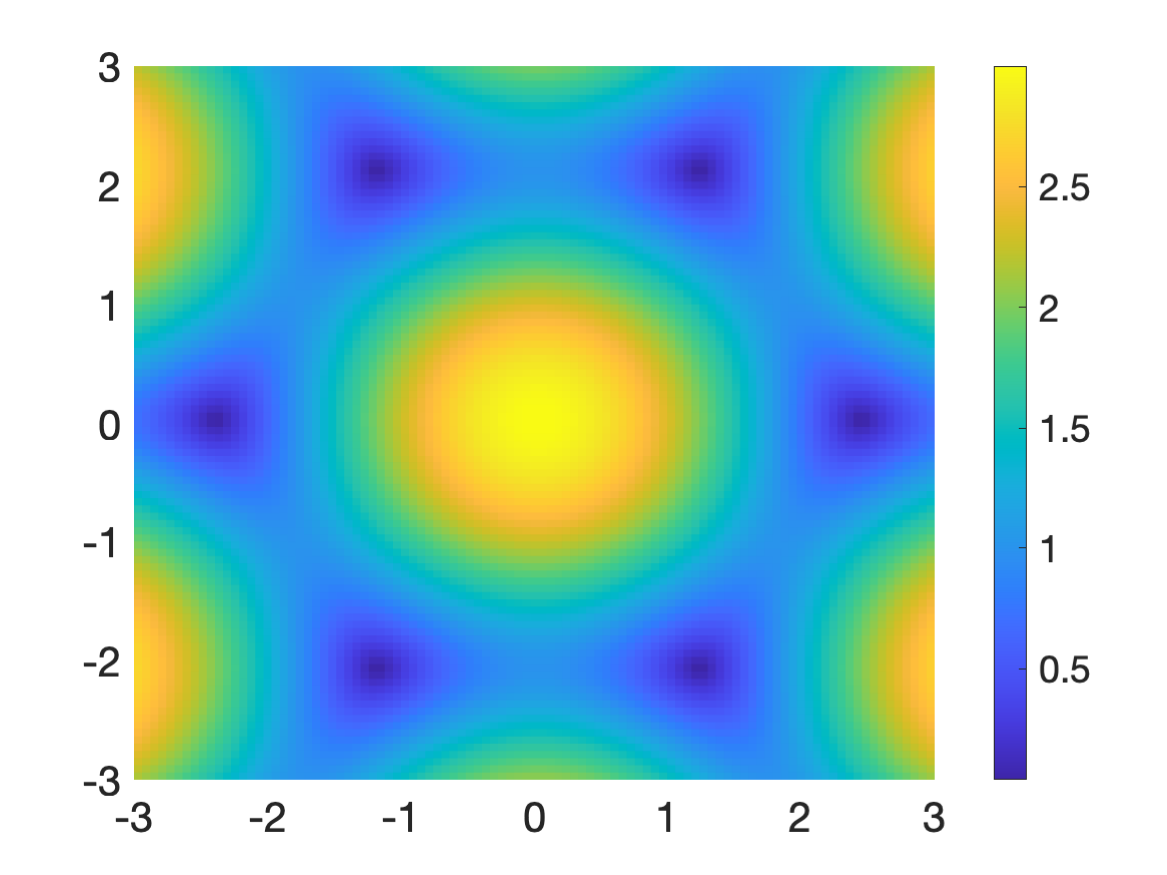}
\includegraphics[width=5.3cm]{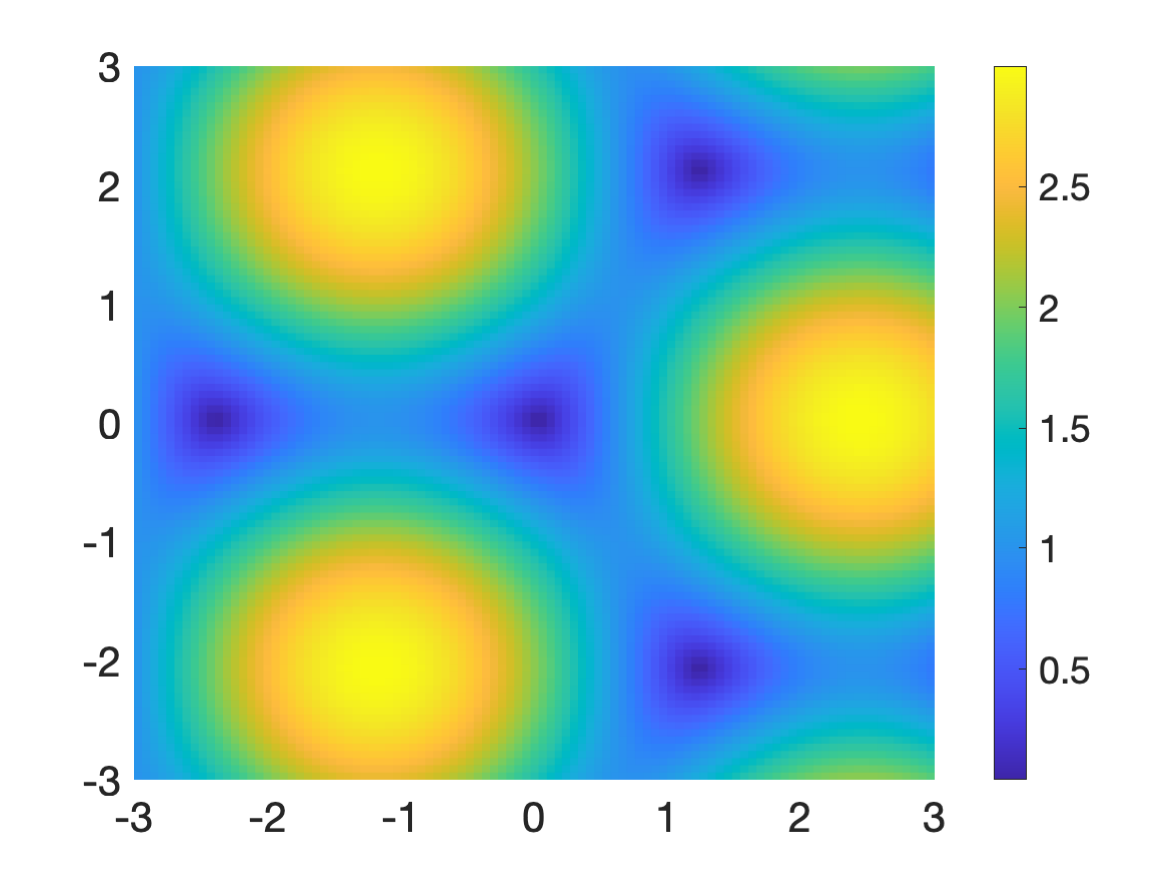}
\includegraphics[width=4.5cm]{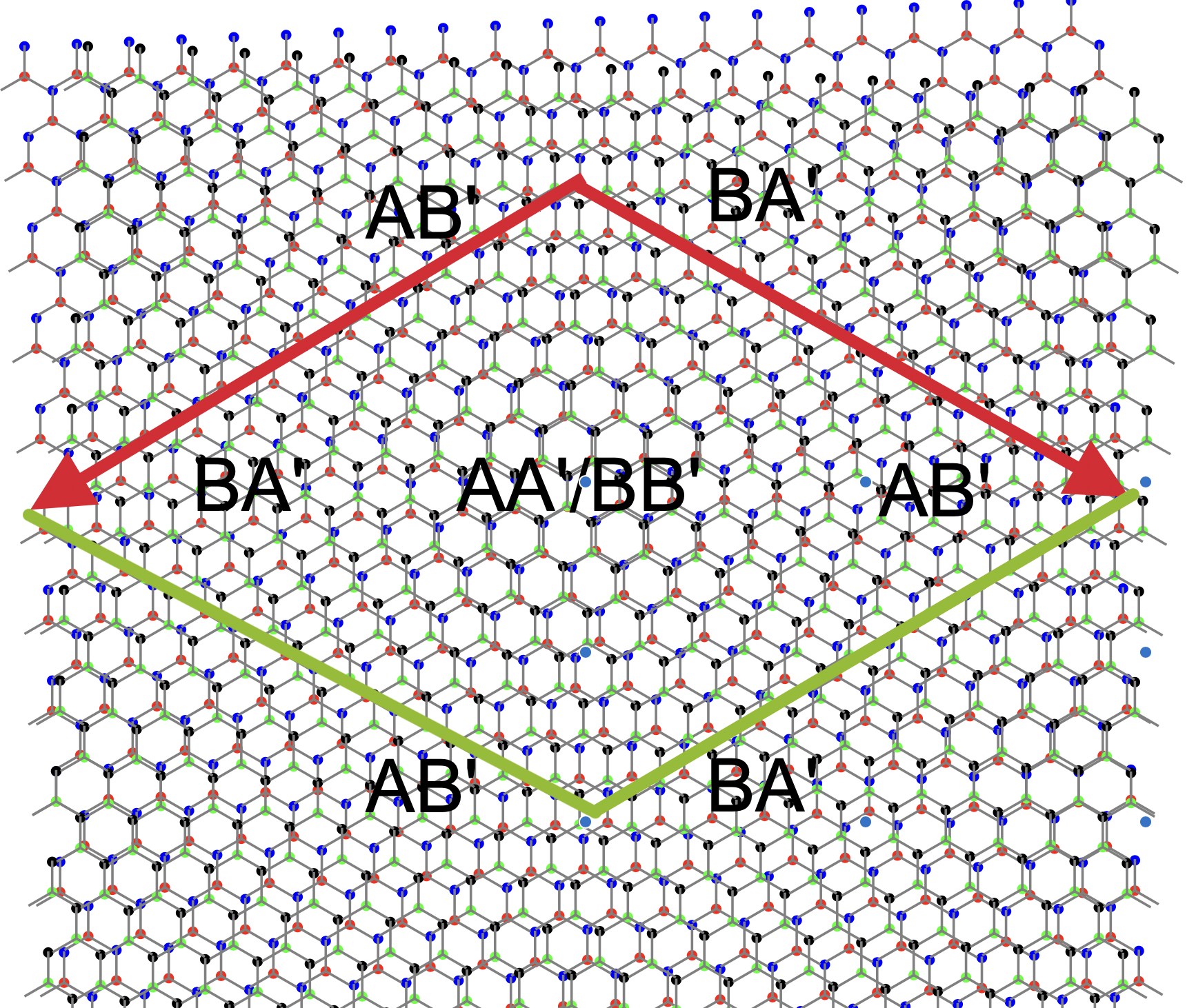}
\caption{\label{fig:tunnelling}$\vert V(\mathbf r)\vert$ (left) peaked at $AA'/BB'$ regions, $\vert U(\mathbf r)\vert$ (center) concentrated at $AB'$ regions, and figure of moir\'e fundamental cell with regions of particle-type overlap (right).}
\end{figure}

\section{Symmetries of BMH}

In this section we review the symmetries of BMH given in \eqref{eq:defBM} and the
spaces used in the definition of the bands $ k \mapsto E_{\pm \ell } ( \alpha, \lambda, k ) $, 
$ \ell \in \mathbb Z \setminus \{ 0 \} $. We follow the presentation of \cite[Supplementary Materials]{phys},
with modifications due to the variables used here, see \cite[Appendix A]{bhz2}. 

\subsection{Translation and rotation symmetries}
The translation symmetry for $ u \in L^2_{\rm{loc}} ( \mathbb C ; \mathbb C^2 ) $ is given by 
\begin{equation}
\label{eq:defLag}  
L_{\gamma } u  :=  
\begin{pmatrix} e^{ i \langle \gamma, K \rangle }   & 0  \\
0 & e^{-  i \langle \gamma , K \rangle } 
\end{pmatrix}  u ( z + \gamma )   ,   \ \ \ 
\gamma \in \Lambda,  \ \  K = \tfrac43 \pi .  \end{equation}
We extend this action diagonally for $ w \in L^2_{\rm{loc}} ( \mathbb C ; \mathbb C^4 ) $:
\[  \mathscr L_\gamma w = \begin{pmatrix} L_\gamma w_1 \\L_\gamma w_2 \end{pmatrix}, \ \ \ 
w = \begin{pmatrix} w_1 \\ w_2 \end{pmatrix}, \ \ w_j \in L^2_{\rm{loc} } ( \mathbb C; \mathbb C^2 ). \]
We then have have, in the notation of \eqref{eq:defBM}, \eqref{eq:Hamiltonian} with $ U $, $ V$
satisfying \eqref{eq:defU} and \eqref{eq:defV}, 
$ L_\gamma D ( \alpha ) = D ( \alpha ) L_\gamma $, and 
$  \mathscr L_\gamma H ( \alpha, \lambda  ) = H ( \alpha ) \mathscr L_\gamma $.
We also define
\begin{equation}
\label{eq:OC} 
\begin{gathered} 
 \Omega : L^2_{\rm{loc}} ( \mathbb C; \mathbb C^2) \to L^2_{\rm{loc}} ( \mathbb C; \mathbb C^2) , \ \ \ 
 \mathscr C :  L^2_{\rm{loc}} ( \mathbb C; \mathbb C^4) \to L^2_{\rm{loc}} ( \mathbb C; \mathbb C^4) ,\\
 \Omega u ( z ) := u ( \omega z ) , \ \  \mathscr C \begin{pmatrix} w_1 \\ w_2 \end{pmatrix} :=
 \begin{pmatrix} \Omega w_1 \\ \bar \omega \Omega w_2 \end{pmatrix}, \end{gathered} \end{equation} 
 so that
 $ \Omega D( \alpha ) = \omega D ( \alpha ) $ and $\mathscr C H( \alpha, \lambda ) = 
 H ( \alpha, \lambda ) \mathscr C $. 
 
 The natural subspaces of $ L^2_{\rm{loc}} ( \mathbb C ; \mathbb C^p ) $, $ p = 2, 4 $, are given by
 \begin{equation}
 \label{eq:Lk} 
 \begin{gathered}
 L^2_k ( \mathbb C ; \mathbb C^4 ) := \{ u \in L^2_{\rm{loc} } ( \mathbb C, \mathbb C^4 ) : 
 \mathscr L_\gamma u = e^{ i \langle k , \gamma\rangle } u \}, 
  \end{gathered}
 \end{equation}
with $  L^2_k ( \mathbb C ; \mathbb C^2 )$ defined by replacing $ \mathscr L_\gamma $ by 
$ L_\gamma$. When the context is clear we simply write $ L^2_k $ for either space. We note
that these spaces depend only on the congruence class of $ k $ in $ \mathbb C/\Lambda^* $, 
\[ \Lambda^* := \frac{ 4 \pi i }{\sqrt 3 } \Lambda , \ \ \   k \mapsto z ( k ) :=\frac{ \sqrt{3}k }{ 4 \pi i } ,
\ \ \ \Lambda^* \to \Lambda, \ \  \langle p , \gamma \rangle \in 2 \pi \mathbb Z,  \ \ p\in \Lambda^*, \ 
\gamma \in \Lambda . \]
The points of high symmetry, $ \mathcal K $, are defined by demanding that
$ \omega p \equiv p \! \mod \! \Lambda^*$. They are given by 
$ \mathcal K = \{ K, - K , 0 \} + \Lambda^* $, $ K = \frac 4 3 \pi $ (see \eqref{eq:defUV}, \eqref{eq:defLag}).
For $ k \in \mathcal K/\Lambda^* $ and $ p \in \mathbb Z_3 $ we also define
\begin{equation}
\label{eq:Lkp} 
L_{k,p}^2 ( \mathbb C ; \mathbb C^4 ) := \{ 
u \in L^2_{k} ( \mathbb C; \mathbb C^4 ) :  \mathscr C^p u = \bar \omega^p u \} , 
\end{equation}
with the definition of $ L_{k,p}^2 ( \mathbb C ; \mathbb C^2 )$ obtained by replacing $\mathscr C $
by $ \Omega $. We have orthogonal decompositions $ L^2_k = \bigoplus_{ p \in \mathbb Z_3}
L_{k,p} $, $ k \in \mathcal K/\Lambda^*$. Also, the actions of
$ \mathscr L_\gamma $ and $ \mathscr C $ on $ L^2_{p,k} $ commute. (In general, 
 $ \mathscr L_\gamma \mathscr C = \mathscr C \mathscr L_{\omega \gamma}$.)

As recalled in \S \ref{s:int}, the Floquet spectrum is obtained by taking the spectrum of
\begin{equation}
\label{eq:defHk}  H_k ( \alpha, \lambda ) := 
e^{ -i \langle z, k \rangle} H ( \alpha, \lambda ) e^{ i \langle z, k \rangle } = 
\begin{pmatrix} \lambda C & D(\alpha) ^* + \bar k \\
D( \alpha) +k & \lambda C \end{pmatrix} , \end{equation}
on $ L^2_0 $ with the domain given by $ H^1_{\rm{loc}} \cap L^2_0 $. 

The spaces $ L^2_{k,p} $ will be crucial in showing that $ f ( \alpha , \pm K ) = 0 $ in \eqref{eq:Epm1}.
That, and the proof of \eqref{eq:prope}, requires additional symmetries, well known in the
physics literature.
 
\subsection{Additional symmetries}

We start with the symmetries of $ D( \alpha ) $, a non-self-adjoint building block of BMH \eqref{eq:Hamiltonian} for $ U $ satisfying the conditions in \eqref{eq:defU}.
We start with 
\begin{equation}
\label{eq:defE}
\mathscr E D ( \alpha ) \mathscr E^{-1}  = - D( \alpha ) , \ \ \mathscr Ev(z):=Jv(-z), \quad 
J:=\begin{pmatrix} \ \ 0 & 1 \\ -1 & 0 \end{pmatrix}, \ \
\mathscr E : L^2_{0} \to L^2_{0} . 
 \end{equation}
If $ \alpha \in \mathcal A $ is simple then the kernel, $ \mathbb C u_0 $, of $ D (\alpha ) $ on $ L^2_0 $ is one dimensional
(see \cite[Theorem 2]{bhz2}) and, as the spectrum of $ \mathscr E  $ given by 
$ \{ \pm i \} $ ($ \mathscr E^2 = - I $), $ \mathscr E u_0 = \pm i 
 u_0 $ (for one choice of sign). Hence, 
\begin{equation}
\label{eq:phi2psi}     u_0 ( z ) = \begin{pmatrix} \psi ( z ) \\
\pm i \psi ( - z )\end{pmatrix} , \ \ \ \psi ( z + \gamma ) = e^{ -i \langle \gamma, K \rangle } \psi ( z ), \ \ 
\gamma \in \Lambda. 
\end{equation}
(This observation about the Bloch wave function associated to the $ \Gamma $ point
was made in -- see the discussion before \cite[(26)]{magic}.)
Interestingly,  numerical evidence suggests that the signs alternate, with $ - $ at the first real magic angle, then $+$, and so on. For future reference we also recall from \cite[Proposition 3.6]{dynip} that
$ u_0 \in L^2_{0,2} $, that is
\begin{equation}
\label{eq:rotpsi}     \psi ( \omega z ) = \omega \psi ( z ) , \ \ \psi ( z ) = z \psi_1 (z ), \ \ \psi_1 ( 0 ) \neq 0  \text{ with } \psi_1 \in C^{\omega}(\CC;\CC).
\end{equation}

We next recall an anti-linear symmetry,
\begin{equation}
 \label{eq:defQ} {\mathscr Q}v(z) := \overline{ v(-z)}
 , \ \ \ \ {\mathscr Q}(D(\alpha) + k){\mathscr Q} = (D(\alpha) + k)^*, 
 \end{equation}
and two linear symmetries,
\begin{equation}
\label{eq:defH}  \mathscr H \begin{pmatrix} u_1 ( z )\\u_2 ( z ) \end{pmatrix}  :=  \begin{pmatrix} 
- i u_2 ( - z ) \\ i u_1 ( - z) \end{pmatrix} , \ \ \ \   \mathscr H ( D ( \alpha ) + k )  \mathscr H = - ( D( \alpha ) 
- k ) , \end{equation}
(We have, in the notation of \eqref{eq:defE}, $ i \mathscr H = \mathscr E $; the factor $ i $ is there
for $ \mathscr H^2 = I$.)
\begin{equation}
\label{eq:defN}
\mathscr N u ( z ) := \begin{pmatrix} u_2 ( - \bar z ) \\ u_1 ( - \bar z ) \end{pmatrix}, \ \ \ \
\mathscr N ( D ( \alpha ) + k ) \mathscr N =  - ( D ( \bar \alpha ) - \bar k )^* .
\end{equation}

All of these operators satisfy, $ A^2 = I $ and $ A : L^2_0 \to L^2_0 $. Remarkably, for potentials in \eqref{eq:Hamiltonian} satisfying \eqref{eq:defU} and \eqref{eq:defV} these
symmetries give symmetries of BMH (acting 
on $ L^2_{\rm{loc}} ( \mathbb C; \mathbb C^4 ) $). They have the following quaint physical names:
\begin{equation}
\label{eq:quaint}
\begin{aligned} 
& \text{\bf $ \mathcal P\mathcal T $/$C_{2z}T$ symmetry:} \ \ \ &  \mathcal P \mathcal T := 
\begin{pmatrix} 0 & \mathscr Q \\ \mathscr Q & 0 \end{pmatrix}, \\
& \text{\bf particle-hole symmetry:} \ \ \ &  \mathscr S := \begin{pmatrix} \mathscr H & 0 \\
0 & \mathscr H \end{pmatrix}, 
\\ & \text{\bf mirror symmetry:} \ \ \  &  \mathscr M := \begin{pmatrix} \ 0 & i \mathscr N \\
-i \mathscr N & 0 \end{pmatrix} .
\end{aligned} 
\end{equation}
These operators all satisfy $ A^2 = I$ but the mapping properties are more complicated,
see \eqref{eq:mapp} below. 
The commuting components of the $ \mathcal P \mathcal T $ symmetry are given by 
the parity-inversion and time-reversal operators:
\[ \mathcal P  u := \begin{pmatrix} 0 & I_{\CC^2} \\
\ I_{\CC^2} & 0 \end{pmatrix}  u ( - z ) , \ \ \ \ \mathcal T  u ( z ) := \overline {u ( z )} .\]
We also remark that the antilinear symmetry $ \mathscr A $ which holds for more general $ D ( \alpha ) $
is given by $ \mathscr A = - i \mathscr H \mathscr Q $, see \cite[\S 6.1]{bhz3} and \cite[\S 3.5]{dynip}.

We summarize the basic properties of these symmetries in the following
\begin{prop}
\label{p:sym}
With the definitions given in \eqref{eq:defBM},\eqref{eq:quaint} and for potentials
satisfying \eqref{eq:defU}, \eqref{eq:defV},  we have, for $ \alpha \in \mathbb C $, $ \lambda \in \mathbb R $, 
\begin{equation}
\label{eq:comm}
\begin{gathered} 
\mathcal P \mathcal T H ( \alpha, \lambda ) = H (  \alpha, \lambda ) \mathcal P \mathcal T , \ \ \ \  
\mathscr M H ( \alpha, \lambda) = H ( \bar \alpha, \lambda ) \mathscr M, \\
\mathscr S H ( \alpha, \lambda ) = - H ( \alpha, \lambda ) \mathscr S.
\end{gathered} 
\end{equation}
In addition, in the notation of \eqref{eq:Lkp}, and for $ k = \{0, \pm K\} + \Lambda^*$, $ p \in \mathbb Z_3$, 
$ K = \frac43 \pi $, 
\begin{equation}
\label{eq:mapp}
\mathcal P \mathcal T : L^2_{k,p} \to L^2_{k,-p+1}, \ \ 
\mathscr S : L^2_{k,p} \to L^2_{k+K,p}, \ \ \mathscr M : L^2_{k,p} \to L^2_{-k, 1-p}, 
\end{equation}
where  $L^2_{\bullet, \bullet} = L^2_{\bullet, \bullet} ( \mathbb C; \mathbb C^4 ) $.
\end{prop}
\begin{proof}
To establish \eqref{eq:comm} we start with the $ \mathcal P \mathcal T $ symmetry where in view
of \eqref{eq:defQ} we have 
\[ \begin{pmatrix} 0 & \mathscr Q \\ \mathscr Q & 0 \end{pmatrix} 
\begin{pmatrix} 
\lambda C & D ( \alpha)^* \\ D ( \alpha ) & \lambda C \end{pmatrix} = 
\begin{pmatrix}  \lambda \mathscr Q  C \mathscr Q &  D ( \alpha)^* \\
 D ( \alpha ) & \lambda \mathscr Q  C \mathscr Q \end{pmatrix} \begin{pmatrix} 0 & \mathscr Q \\ \mathscr Q & 0 \end{pmatrix}.
 \]
Since $ \overline{ V ( -z )} = V ( z ) $ (see \eqref{eq:defV}) we have $ \mathscr Q C \mathscr Q = C $. 

For the mirror symmetry in \eqref{eq:quaint} we use \eqref{eq:defN} to see that
\[ \begin{pmatrix} \ 0 &  i \mathscr N \\
-i \mathscr N & 0 \end{pmatrix} 
\begin{pmatrix} 
\lambda C & D ( \alpha)^* \\ D ( \alpha ) & \lambda C \end{pmatrix} = 
 \begin{pmatrix}  - \mathscr \lambda \mathscr N  C  \mathscr N  &\ \  D ( \bar \alpha)^*  \\
D ( \bar \alpha ) & - \lambda \mathscr N  C \mathscr N  \end{pmatrix} \begin{pmatrix} \ 0 &  i \mathscr N \\
-i \mathscr N & 0 \end{pmatrix} .  \]
From $ V ( \bar z ) = V ( z ) $ (see \eqref{eq:defV}, we conclude
$ \mathscr N  C \mathscr N = C $, proving $ \mathscr M H ( \alpha, \lambda ) = 
H ( \bar \alpha, \lambda ) \mathscr M $. 
Finally, $  \mathscr S H ( \alpha, \lambda ) = - 
H ( \bar \alpha, \lambda ) \mathscr S $ follows from \eqref{eq:defH} and 
$ \mathscr H C \mathscr H = - C $. 

To obtain \eqref{eq:mapp}, we first note that 
 $ L_\gamma  \mathscr Q  =  \mathscr Q L_{-\gamma} $, 
$ L_\gamma \mathscr H = e^{ i \langle \gamma, K \rangle} \mathscr H L_{-\gamma } $ and
$ L_\gamma \mathscr N =  L_{-\bar \gamma} \mathscr N$. Hence, 
\begin{equation}
\label{eq:comL}  \mathscr L_\gamma \mathcal P \mathcal T  =  \mathcal P \mathcal T \mathscr L_{-\gamma }, \ \ \ \
\mathscr L_\gamma \mathscr S = e^{ i \langle K, \gamma \rangle } \mathscr S \mathscr L_\gamma,  \ \ \ \
\mathscr L_\gamma  \mathscr M = 
\mathscr M \mathscr L_{-\bar \gamma} .\end{equation}
In the notation of \eqref{eq:OC}, $ \Omega \mathscr Q = \mathscr Q \Omega $, $ \Omega \mathscr H = \Omega  \mathscr H $ and $ \Omega \mathscr N = \mathscr N \Omega^* $ and thus
\begin{equation}
\label{eq:comC}
\mathscr C \mathcal P \mathcal T = \bar \omega \mathcal P T \mathscr C, \ \ \ \ 
\mathscr C \mathscr S = \mathscr S \mathscr C, \ \ \ \
\mathscr C \mathscr M = \bar \omega \mathscr M \mathscr C^* . 
\end{equation}
We can now check the mapping properties \eqref{eq:mapp}: for $ u \in L^2_{p,k} $, 
\[ \mathscr L_\gamma \mathscr C  \mathcal P \mathcal T u = 
\bar \omega \mathcal P \mathcal T \mathscr C \mathscr L_{-\gamma } u =
\bar \omega \mathcal P \mathcal T (\bar \omega^p e^{ - i \langle \gamma, k \rangle } u )
= e^{ i \langle \gamma , k \rangle } \bar \omega^{1-p} \mathcal P T u , \]
that is, $ \mathcal P \mathcal T u \in L^2_{k,1-p} $. Then
\[ \mathscr L_\gamma \mathscr C  \mathscr S  u = 
e^{ i \langle K, \gamma\rangle } \mathscr S \mathscr C \mathscr L_\gamma u = 
e^{ i \langle K +k , \gamma \rangle } \bar \omega^p \mathscr S u , \]
that is $ \mathscr S u \in L^2_{k+K, p} $. Finally, 
\[ \mathscr L_\gamma \mathscr C  \mathscr M  u = 
\bar \omega \mathscr M \mathscr C^* \mathscr L_{-\bar \gamma} u = 
\bar \omega^{1-p}  e^{ - i \langle \gamma, k \rangle } \mathscr M u, \]
(since $ k = \pm K ,  0 $, $ \langle \gamma, k \rangle = 
\langle \bar \gamma, k \rangle$) which means that $ \mathscr M u \in L^2_{-k,1-p} $. This completes the proof of 
\eqref{eq:mapp}. 
\end{proof}

\subsection{Protected states of BMH}

We now use a modification of the argument from \cite{phys} to obtain symmetry 
protected states for $ H ( \alpha, \lambda ) $. In the case of the chiral limit $ H ( \alpha, 0) $, 
the same conclusion is easier and is based on the more immediate symmetry of
eigenvalues of $ H ( \alpha, 0 ) $ about $ 0 $ \cite{magic,beta}.

\begin{prop}
\label{p:prote}
In the notation of \eqref{eq:defHk} and for all $ \alpha, \lambda \in \mathbb R $, 
\begin{equation}
\label{eq:prote}
\dim \ker_{ L^2_0 }  H_{\pm K} ( \alpha , \lambda )  \geq 2 , \ \ \ K := \tfrac43 \pi,
\end{equation}
$ \omega K \! \equiv K \!\! \mod \! \Lambda^*$.
\end{prop}
\begin{proof}
Since $ H_{\pm K } ( \alpha, \lambda ) = e^{\mp i \langle z , K \rangle} H ( \alpha , \lambda ) e^{\pm  i \langle 
z , k \rangle } $, and multiplication by $ e^{ \pm i \langle z, K \rangle } $ takes $ L^2_0 $ to $ L^2_{\pm K} $,
\eqref{eq:prote} is equivalent to $ \dim \ker_{L^2_{\pm K } } H ( \alpha, \lambda ) \geq 2 $, for 
all real $ \alpha $ and $ \lambda $. 
We now observe that Proposition \ref{p:prote} gives 
\begin{equation}
\label{eq:Upm} \begin{split}  & U_+ H ( \alpha, \lambda ) 
U_+^{-1}  = - H ( \alpha, \lambda ) , \ \ \ U_+ :=  \mathscr S \mathscr M \mathcal P \mathcal T, \\
& 
U_- H ( \alpha, \lambda ) U_-^{-1}  
= - H ( \alpha, \lambda ) , \ \ \ U_- := \mathscr M \mathscr S  \mathcal P \mathcal T .
\end{split}  \end{equation}
Also
\begin{equation}
\label{eq:Upm1} U_\pm : L^2_{ \pm K, p } \to L^2_{\pm K, p } .\end{equation}
In fact, from \eqref{eq:mapp} we see that
\[ \begin{split} & L^2 _{K,p} \xrightarrow{ \mathscr S } L^2_{2K,p} \xrightarrow{ \mathscr M } L^2_{-2K,1-p}  \xrightarrow{ \mathcal P \mathcal T } L^2_{-2K,p} = L^2_{K,p} \\
& L^2_{-K,p} \xrightarrow{ \mathscr M } L^2_{K,1-p} \xrightarrow{ \mathscr S } L^2_{2K,1-p}  \xrightarrow{ \mathcal P \mathcal T } L^2_{2K,p} = L^2_{-K,p} .
\end{split} 
\]
The mapping property \eqref{eq:Upm1} and \eqref{eq:Upm} show that 
\begin{equation}
\label{eq:evens} \Spec_{ L^2_{\pm K, p} } H( \alpha, \lambda )  = -  \Spec_{ L^2_{\pm K, p} } H( \alpha, \lambda ). \end{equation}
If $ \{ \mathbf e_j \}_{j=1}^4 $ is the standard basis of $ \mathbb C^4 $ then 
 we easily check that
\begin{gather*}  \ker_{ L^2_{ K, 0 } } H ( 0, 0 ) = \mathbb C \mathbf e_1 , 
\ \ \ \  \ker_{ L^2_{ -K, 0 } } H ( 0, 0 ) = \mathbb C \mathbf e_2 ,\\ 
 \ker_{ L^2_{ K, 1 } } H ( 0, 0 ) = \mathbb C \mathbf e_3 ,\ \ \ \ 
  \ker_{ L^2_{ -K, 1 } } H ( 0, 0 ) = \mathbb C \mathbf e_4 ,
  \end{gather*} 
 and in particular all these spaces are one dimensional.
 In view of \eqref{eq:evens}, these dimensions have to remain odd for all $ \alpha, \lambda \in 
 \mathbb R $ and that shows that $ \dim \ker_{L^2_{\pm K } } H ( \alpha, \lambda ) \geq 2 $, proving
 \eqref{eq:prote}. 
\end{proof}

\section{Perturbation theory} 
\label{s:pt}

Suppose that $ \alpha \in \mathcal A \cap \mathbb R $ is 
simple (see \S \ref{s:int}). Then by \cite[Theorem 2]{bhz2},  $ \ker ( D ( \alpha)  + k ) = \mathbb C u( k ) $,
$ u ( k ) \neq 0 $ for all $ k \in \mathbb C $. We also have $ \ker ( D^*( \alpha )  + \bar k ) = \mathbb C u^* ( k ) $ 
and the fact that the kernels are nonempty for all $ k $ implies that
$\langle u ( k ) | u^* ( k ) \rangle = 0$. 
(Otherwise perturbation theory, as in the Grushin problem below, would give us invertibility of
 $  D ( \alpha ) + k $.) We normalize $ u(k) , u^*(k)  $ to have norm one.
To simplify notation we temporarily drop $ \alpha $ in $ D ( \alpha ) $ and consider $ H_k ( \alpha, 0 ) $, a self-adjoint Hamiltonian corresponding to the chiral limit:
 \[  H ( k )   := \begin{pmatrix} 0 & D^* + \bar k  \\
 D + k & 0 \end{pmatrix}.  \] 
We then define
\[  \begin{gathered} \mathcal H ( k ) := \begin{pmatrix} H ( k ) & R_- ( k ) \\
R_+ ( k ) & 0 \end{pmatrix} ,   \ \ \ R_- ( k ) : \mathbb C^2 \to L^2_0 , \ \ \ 
R_+ ( k ) = R_- ( k )^* : L^2_0 \to \mathbb C^2 , \\
  R_- ( k ) := \begin{pmatrix} 0 & | u ( k ) \rangle \\
| u^* ( k ) \rangle & 0 \end{pmatrix} \ \ \ 
R_+ ( k ) = \begin{pmatrix} 0 & \langle u^* ( k ) | \\
\langle u ( k )|  & 0 \end{pmatrix} , \end{gathered} \]
so that 
\begin{equation}
\label{eq:Einv}
\begin{gathered}
\mathcal  H (k ) ^{-1} = \begin{pmatrix} E ( k ) & E_+ ( k ) \\
E_- ( k ) & E_{-+} ( k ) \end{pmatrix} , \ \  E_+ ( k ) = R_- ( k ) , \ \ E_- ( k ) = R_+ ( k )\\
E ( k )  := \begin{pmatrix} 0 & E_0(k) ^* \\ E_0(k) & 0 \end{pmatrix}, \ \ \ E_{-+ } ( k ) \equiv 0 , \\
E_0( k ) := 
\left( ( D (\alpha ) + k )|_{  |u( k )\rangle^\perp \to |  u^*(k)\rangle^\perp} \right)^{-1}
( I -  | u^* ( k ) \rangle \langle u^* ( k )|  ) . \end{gathered} \end{equation} 
We label the eigenvalues of $ H(k)  $ as $ E_{\ell} ( k )  $, $ \ell \in \mathbb Z \setminus \{ 0 \} $, $ E_{\ell} ( k)  = - E_{-\ell} ( k )  $. Under our assumptions we have $ E_{\pm 1 } ( k )  = 0 $, $ E_{\ell } ( k) \neq 0$ $ |\ell|> 1$, see \cite[Theorem 2]{bhz2}.
We now consider the BMH \eqref{eq:defBM} as a perturbation of $ H ( k ) $:
\[ H( k , \lambda ) := \begin{pmatrix} \lambda C & D^* + \bar k \\
D+k & \lambda C \end{pmatrix} , \ \ C = C^* . \]
For $ |\lambda| \ll 1  $, the two smallest (in modulus) eigenvalues of $ H ( k, \lambda ) $, $ E_{\pm 1 } ( k , \lambda ) $, are given by 
 the eigenvalues of 
\[ E_{-+} ( k, \lambda ) = - \sum_{\ell=0}^\infty (-1)^\ell  \lambda^{\ell+1} E_-(k) C(E(k)C)^{\ell}E_{+} ( k ) : \mathbb C^2 \to \mathbb C^2 , \]
see \cite[Proposition 2.12]{notes} for this standard fact. 
Using \eqref{eq:Einv} we then 
obtain
\begin{equation}
\label{eq:Eexp}  \begin{gathered} E_{-+} ( k , \lambda ) = 
\begin{pmatrix} \lambda e_{+} ( k ) & \lambda^2 f ( k ) \\
 \lambda^2 \overline{f ( k ) } & e_- ( k ) \end{pmatrix}  + \mathcal O ( \lambda^3)_{ \mathbb C^2 \to 
\mathbb C^2 } , \\
e_+  ( k ) = - \langle u ( k ) | C | u ( k ) \rangle , \ \ \ 
e_-  ( k ) = - \langle u^* ( k ) | C | u^* ( k ) \rangle , 
\\ f ( k ) = \langle u ( k ) | \bar V E_0(k)^* \bar V | u ( k ) \rangle . 
\end{gathered}
\end{equation}

We should stress (see \cite[\S 5.2]{bhz2}) that $ u ( k ) $ cannot be made continuous in $ k $ and 
$ \Lambda^* $-periodic (even modulo a $k$-independent unitary transformation, 
see \eqref{eq:tau} below).  However $ E_{-+} ( k , \lambda ) $ is periodic. We present a quick 
argument relevant to \eqref{eq:Epm1}. 

\begin{lemm}
\label{l:per}
In the notation of \eqref{eq:Eexp}, $ e_\pm, f \in C^\infty ( \mathbb C ) $ and
\begin{equation}
\label{eq:per}
e_\pm ( k + p ) = e_\pm ( k ) , \ \  f ( k + p ) = f ( k ) , \ \ \ k \in \mathbb C, \ \ p \in \Lambda^* .
\end{equation}
\end{lemm}
\begin{proof}
We recall that
\begin{equation}
\label{eq:tau}
\begin{gathered} 
\tau ( p )^* ( D ( \alpha ) + k ) \tau ( p ) = D ( \alpha ) + k + p , \\ 
\tau ( p ) v ( z ) := e^{ i \langle z, p \rangle } v ( z ) ,  \ \  p \in \Lambda^* , \ \ 
\tau ( p ) : H^s ( \mathbb C/\Lambda ) \to H^s ( \mathbb C/\Lambda ), \ \ s \in \mathbb R . 
\end{gathered} 
\end{equation}
We know (see \eqref{eq:uke} for an explicit formula) that $ k \mapsto u ( k) $ is a
smooth function of $ k \in \mathbb C $.
Simplicity of $ \alpha $ then shows that
\[  | u ( k + p ) \rangle = \rho_p ( k ) \tau ( p )^* u ( k ) \rangle, \ \ \ | \rho_p ( k ) | = 1, 
\ \ \ k \in \mathbb C , \ \ p \in \Lambda^* . 
 \]
(We have an explicit formula for $ \rho_p ( k ) $ coming from \eqref{eq:uke}, see
\cite[\S 5]{bhz2}, but that is of no importance here.) This and \eqref{eq:tau} also 
shows that $ \tau ( p )^* : | u ( k ) \rangle^ \perp \to | u ( k + p ) \rangle ^\perp $, 
$ \tau(p) : | u^*( k )\rangle^\perp \to  | u^* ( k + p ) \rangle ^\perp $ and hence, 
in the notation of \eqref{eq:Einv}, 
\[ \begin{gathered}   \langle u ( k + p ) | = \overline{ \rho_p ( k )} \langle u ( k ) | \tau ( p ) |, \ \ \
| u ( k + p ) \rangle = \rho_p ( k ) | \tau ( p )^* | u ( k ) \rangle , \\
 E_0 ( k + p ) = \tau ( p ) E_0 (k ) \tau ( p)^* ,  \ \ \ 
 k \in \mathbb C ,   \  p \in \Lambda^* . \end{gathered} \]
Since $ C $ and $ \bar V $ are multiplication operators which commute with $ \tau ( p ) $ 
and $ \tau( p )^* $, \eqref{eq:per} follows.
\end{proof} 

We also have
\begin{lemm}
\label{l:p2m}
In the notation of \eqref{eq:Eexp} we have 
\[  e_+ ( k ) = e_- ( k ) =: e ( k ) . \]
\end{lemm}
\begin{proof} 
We calculate
\begin{equation}
\label{eq:ep1}   e_+ ( k )  = - \langle u ( k ) | C | u ( k ) \rangle 
  = 
- \langle u_1 ( k ) | V | u_2 ( k ) \rangle - \langle u_2 ( k ) | \bar V | u_1 ( k ) \rangle .
\end{equation} 
In view of \eqref{eq:defQ} we can take
\[ u^* ( k ) = {\mathscr Q} u ( k ) = 
\begin{pmatrix} \overline{ u_1 (k , - z ) } \\ \overline{u_2 ( k , - z ) } \end{pmatrix} , \]
so that the definition of $ e_- ( k ) $ gives,
(recall from \eqref{eq:defV} that  $V ( - z ) = \overline{V ( z ) } $) 
\[  \begin{split} e_-( k ) & = - \langle u^* ( k ) | C | u^* ( k ) \rangle = 
- \langle u_1^* ( k ) | V | u_2^* ( k ) \rangle - \langle u_2^* ( k ) | \bar V | u_1^* ( k ) \rangle 
\\ & = - \langle \overline{ u_1 ( k , - \bullet ) } | V | \overline { u_2 ( k, - \bullet)} - 
\langle \overline{ u_2 ( k , - \bullet ) } | \bar V | \overline{ u_1 ( k , - \bullet ) } \rangle
 \\ & = - \langle u_2 ( k )  | V ( - \bullet) | u_1 ( k) \rangle - 
  \langle  u_1 ( k  ) | \overline{ V (- \bullet )}  | u_2 ( k ) \rangle \\
  & =  - \langle u_2 ( k )  | \bar V  | u_1 ( k) \rangle - 
  \langle  u_1 ( k  ) | V   | u_2 ( k ) \rangle = e_+ ( k ) ,
\end{split} 
\]
which is the desired conclusion. \end{proof} 

Hence the eigenvalues of \eqref{eq:Eexp} are given by \eqref{eq:Epm1}. From 
Proposition \ref{p:prote} we see that $ E_{\pm } ( \alpha, \lambda, \pm K ) = 0 $
and hence $ e(\pm K)= f ( \pm K ) = 0 $. It remains to establish \eqref{eq:prope}.
This will be done in \S \ref{s:sym}.

We conclude this section with an expression for $ e( k) $ in terms of $ \psi $ 
in \eqref{eq:phi2psi}. To that we recall the notation of \cite{dynip}:
\begin{equation}
\label{eq:defFk}  F_k ( z ) = e^{\frac i 2   (  z -  \bar z ) k } \frac{ \theta ( z - z ( k ) ) }{
\theta( z) } ,  \ \ \ z (k):=  \frac{ \sqrt 3 k }{ 4 \pi i } , \ \  z:  \Lambda^* \to \Lambda ,
\end{equation}
where $ \theta ( \zeta ) = \theta_1 ( \zeta| \omega ) $ is a Jacobi theta function
-- see \cite[\S 2.2]{voca}.
We recall that 
\begin{equation}
\label{eq:propFk}
\begin{gathered}F_k ( z + m + n \omega ) =  F_k ( z ) , \\
  ( 2 D_{\bar z } + k ) F_k ( z ) 
 = a(k)   \delta_0 ( z ) , \ \ a(k) :=  2 \pi i {\theta ( z ( k ) ) }/{ \theta' ( 0 ) } .
\end{gathered}
\end{equation}
We can then use \eqref{eq:phi2psi} and take
\begin{equation}
\label{eq:uke}  u ( k ) = c( k) F_k ( z ) u_0 ( z) =  c( k) F_k ( z ) 
\begin{pmatrix} \ \ \psi ( z ) \\ \pm i \psi ( - z ) \end{pmatrix} , \end{equation}
where $ c ( k ) > 0 $ is the normalizing constant. This gives
\begin{equation}
\label{eq:epk}  \begin{split} e_+ ( k ) & = - \langle u ( k ) | C | u ( k ) \rangle 
 = 
- \langle u_1 ( k ) | V | u_2 ( k ) \rangle - \langle u_2 ( k ) | \bar V | u_1 ( k ) \rangle \\
& = - c(k)^2 \int_{ \mathbb C/\Lambda } ( V ( z) u_2 ( k , z ) \overline{ u_1 ( k , z ) } 
+ \overline{V ( z ) } u_1 ( k, z ) \overline{u_2 ( k , z ) } dm ( z ) \\
 & = \mp i c (k)^2 \int_{\mathbb C/\Lambda } | F_k ( z ) |^2 ( V ( z ) \psi ( -z ) \overline{ \psi ( z ) }
- \overline { V ( z) } \overline { \psi ( -z ) } \psi ( z ) ) dm ( z) \\
& = \pm 2 c ( k )^2 \int_{\mathbb C/\Lambda } | F_k ( z ) |^2 \Im (  V ( z ) \psi ( -z ) \overline{ \psi ( z ) } ) 
dm ( z ) . \end{split} \end{equation}

\section{Symmetries of $ e ( k ) $}
\label{s:sym}

We will now use different symmetries of $ u_0 $ and $ F_k ( z ) $ to obtain
symmetries of $ e (k ) $ listed in \eqref{eq:prope}.

\subsection{The reflection $ k \mapsto - k $ symmetry}
To see that $ e $ is an odd function of $ k$ we use \eqref{eq:defH}: 
from the simplicity assumption $ u (-  k ) = \rho( k ) \mathscr H u ( k ) $, $ | \rho ( k ) | =1 $, and, as $ V ( - \bullet ) = \bar V $, 
\begin{equation}
\label{eq:odd}  \begin{split} e ( - k ) & = 
- \langle u_1 (- k ) | V | u_2 ( - k ) \rangle {\color{green}-}  \langle u_2 ( -k ) | \bar V | u_1 ( - k ) \rangle \\
& = \langle u_2 ( k, - \bullet ) | V | u_1 ( k , -\bullet ) \rangle + \langle u_1( k, - \bullet ) | \bar V | u_2 (  k, 
-\bullet ) \rangle \\
& = \langle u_2 ( k,  ) | V ( -\bullet)  | u_1 ( k  ) \rangle + \langle u_1( k ) | \bar V ( -\bullet )  | u_2 (  k) 
 \rangle \\
& =  \langle u_2 ( k,  ) | \bar V  | u_1 ( k  ) \rangle + \langle u_1( k ) |  V   | u_2 (  k)  \rangle \\
& = - e ( k ) . \end{split} \end{equation} 

\subsection{The rotation $ k \mapsto \omega k $ symmetry}

Using the definition \eqref{eq:defFk}, we consider $ F_k ( \omega z ) $:
\[ a ( k ) \delta_0 ( z) = [( 2 D_{\bar z } + k ) F_k] ( \omega z )  = \omega( ( 2 D_{\bar z }    + 
\bar \omega k  ) [F_k ( \omega z )] .\]
The uniqueness of the fundamental solution of $  2 D_{\bar z } + k $ on the torus shows that
\[   \omega a ( k)^{-1} F_k ( \omega z ) = a ( \bar \omega k )^{-1} F_{\bar \omega k } ( z ) , \]
that is 
\[ F_k ( \omega z ) = 
\frac{ \bar \omega a ( k ) } { a ( \bar \omega k ) } F_{\bar \omega k } ( z ) = 
 \frac{ \bar \omega \theta (  z ( k ) )  }{ \theta (  \bar \omega z (  k ) ) } F_{\bar \omega k } ( z ) . \]

Hence, 
\[ \begin{split} c ( \bar \omega k )^{-2} & = \int_{ \mathbb C/\Lambda } |F_{\bar \omega k } ( z ) |^2 ( |\psi ( z )|^2 + |\psi ( - z )|^2 ) dm ( z ) \\
& = \left| \frac{ \theta ( \bar \omega z ( k ) )  }{ \theta (  z (  k ) ) } \right|^2 
\int_{ \mathbb C/\Lambda } |F_{ k } ( \omega z ) |^2 ( |\psi ( z )|^2 + |\psi ( - z )|^2 ) dm ( z ) 
2= \left| \frac{ \theta ( \bar \omega  z ( k ) )  }{ \theta (  z (  k ) ) } \right|^2  c ( k )^{-2} , \end{split}\]
which implies 
$  c( \omega k ) /{ c ( k ) } =| { \theta ( z (  k ) )} /{ \theta ( \omega  z ( k ) )  }  |$. 
Using \eqref{eq:epk} we obtain
\[ \begin{split}   \pm e_\pm  ( k ) & = 2 c(k)^2 \int_{\mathbb C /\Lambda } | F_k ( \omega z ) |^2 \Im ( V ( \omega z ) 
\overline{ \psi ( \omega z ) } \psi ( - \omega z ) ) dm ( z ) \\
& =  2 c(k)^2 \int_{\mathbb C /\Lambda } | F_k ( \omega z ) |^2 \Im ( V (  z ) 
\overline{ \psi (  z ) } \psi ( -  z ) ) dm ( z ) \\
& 
=  2 c(k)^2 \left| \frac{ \theta ( z (k ) } {\theta ( \bar \omega z ( k ))  } \right|^2
  \int_{\mathbb C /\Lambda } | F_{\bar \omega k } ( z ) |^2 \Im ( V (  z ) 
\overline{ \psi (  z ) } \psi ( -  z ) ) dm ( z ) \\
& = \frac{c ( k )^2 | \theta (  z ( k ) )|^2 } { c ( \bar \omega k )^2 |\theta ( \bar \omega z (k ))|^2 } e (\bar \omega k ) = 
\pm e ( \bar  \omega k ) ,
\end{split} \]
or
\begin{equation}
\label{eq:ebark}
e ( \omega k ) =  e  ( k ) . 
\end{equation}

\subsection{The reflection $ k \mapsto \bar k $ symmetry}
We first claim that 
\begin{equation}
\label{eq:Fbark}
  F_{ - \bar k } ( z ) = 
-    \frac{ \theta ( z ( \bar k ) ) } { \overline{ \theta (z ( k)  )}}  \overline {F_{ k } ( \bar z ) } \ 
\ \text{ and } \ 
c ( k ) = c ( \bar k ) .
\end{equation}
To see the first part, we note that  
\[ ( 2 D_{\bar z } - \bar k )[ \overline{F_{k} ( \bar z ) } ] = 
\overline{ ( - 2 D_z -  k ) [ F_{k } ( \bar z ) ] } = - \overline{[ ( 2 D_{\bar z } +  k ) F_{k} ] ( \bar z ) } =
- \overline{a ( k ) } \delta_0 ( z ), \]
so that \eqref{eq:propFk} and the uniqueness of the Green function give the first part of \eqref{eq:Fbark}.
Hence,
\[ \begin{split}  c ( \bar k )^{-2}  & = \left| \frac{ \theta ( z ( \bar k ) ) } { \theta (z ( k)  )}  \right|^2 
 \int_{ \mathbb C/\Lambda } |F_{k  } (\bar z ) |^2 ( |\psi ( \bar z )|^2 + |\psi ( -\bar  z )|^2 ) dm ( z ) \\
&  = \left| \frac{ \theta ( z ( \bar k ) ) } { \theta (z ( k)  )}  \right|^2  c ( k)^{-2} . 
\end{split} \]
The second part of \eqref{eq:Fbark} then follows from the following fact:
\begin{equation}
\label{eq:theta1}
\theta ( z )  = e^{ - i \pi /4 } \, \overline{\theta (  \bar z )  } .
\end{equation}
In fact, we can use the product representation of $ \theta $ (see \cite[\S I.14]{tata} or \cite[(2.10a)]{voca}):
\[  \theta ( z ) = 2 q^{\frac14} \sin \pi z \prod_{n=1}^\infty ( 1 - q^{2n} ) ( 1 - q^{2n} e^{ 2 \pi i z } )
( 1 - q^{2n} e^{ -2 \pi i z } ),  \ \ \ q = e^{ \pi i \omega }  , \ \ \bar q = - q . \]
Since
\[  q^{\frac 14} = e^{ \frac14 \pi i \omega }, \ \ \  \overline{ q^{\frac14} } = 
e^{ - \frac14 i \bar \omega } = e^{ - \frac14 \pi i } e^{ \frac 14 i \omega } = e^{ -\frac14 \pi i } q , \]
we have 
\[ \overline{\theta ( \bar z )} = 
 2 (- q)^{\frac14} \sin \pi z \prod_{n=1}^\infty ( 1 - q^{2n} ) ( 1 - q^{2n} e^{ - 2 \pi i z } )
( 1 - q^{2n} e^{ 2 \pi i z}  ) = e^{ - \frac14 \pi i } \theta ( z ) , 
\] 
as claimed in \eqref{eq:theta1}.

Assuming that $ \alpha \in \mathcal A \cap \mathbb R $ is simple we can take
$  u^* ( k  ) = \mathscr N u ( - \bar k ) $ in the definition of $ e ( k ) = e_- ( k )$ in \eqref{eq:Eexp}: 
This means that 
\[ \begin{split} 
 u_1^* ( k , z ) & = u_2 ( - \bar k , - \bar z ) = \pm i  F_{- \bar k } ( - \bar z ) \psi (  \bar z ), \\
   u_2^* ( k , z ) & = u_1 ( - \bar k , - \bar z ) = F_{- \bar k }( - \bar z ) \psi ( - \bar z ) . 
   \end{split} 
\]
Using \eqref{eq:defV} again we have
\[  \begin{split} 
e ( k ) & = c(k)^2 \int_{\mathbb C/\Lambda } V ( z) u_2^* ( k , z ) \overline { u_1^* ( k , z ) } 
+  V ( - z) u_1^* (k , z ) \overline { u_2^* ( k , z ) } dm ( z )  \\
& =  \mp i c(k)^2 \int_{\mathbb C/\Lambda } | F_{- \bar k } ( - \bar z ) |^2 
\left( V ( z) \psi ( - \bar z ) \overline{ \psi (\bar z ) } - V ( -z ) \psi ( \bar z ) \overline{ \psi ( - \bar z ) } ) \right)
dm ( z ) \\
& = \mp i c(k)^2 \int_{\mathbb C/\Lambda } | F_{- \bar k } ( z ) |^2 
\left( \overline{ V ( z) } \psi (  z ) \overline{ \psi ( - z ) } -  V (  z )  \psi ( - z ) \overline{ \psi (  z ) } ) \right)
dm ( z ) \\
& = c ( k)^2 c ( \bar k)^{-2} e ( - \bar k ) = c ( k)^2 c ( \bar k)^{-2} e ( - \bar k ).
\end{split} 
\] 
We now recall \eqref{eq:odd} and \eqref{eq:Fbark} to see that $ e( k ) = - e ( \bar k )  $. 

\noindent
{\bf Remark.} We could avoid using \eqref{eq:Fbark} by considering (see \eqref{eq:defQ} and
\eqref{eq:defN}) 
\[ \mathscr R := {\mathscr Q} \mathscr N , \ \ \  \mathscr R ( D ( \alpha ) + k ) \mathscr R = - ( D ( \bar \alpha ) - \bar k ) . \]
This means that for a simple $ \alpha \in \mathbb R\cap \mathcal A  $, 
\[  \mathscr R u ( k ) = \rho ( k ) u ( - \bar k ), \ \ \ |\rho ( k ) | = 1 . \]
or that
\[  \overline{ u_2 ( k , \bar z ) } = \rho ( k ) u_1 ( -\bar k , z ) , \ \ 
\overline{ u_1 ( k, \bar z ) } = \rho ( k ) u_2 ( - \bar k , z ) . \]
Using that in the definition of $ e ( k ) $ we obtain $ e ( k ) = e ( - \bar k ) $. 

\section{Comments on degenerate magic angles}
\label{s:double} 

We now consider degenerate magic angles of multiplicity two -- see \cite[(1.7) and Theorem 1]{bhz3}.
The analogue of Figure~\ref{f:2} is shown in Figure \ref{f:22}. For small $\vert \lambda\vert \le 10^{-3}$, we observe that two out of the four flat bands undergo an upward shift in energy, while the other two undergo an equal and opposite downward shift in energy. This is different compared with simple magic angles where all flat bands move together for small $\lambda.$

Another difference is that the bands are only anchored at $K, K', \Gamma$ at zero energy, rather than
on the line through these points, as was the case for simple magic angles. 
For larger perturbations $\vert \lambda\vert \ge 10^{-2}$, the two upper and two lower bands undergo a splitting into separate bands, respectively. In addition, the two Dirac cones at $\Gamma$ start to overlap (center figure in Fig. \ref{f:22}). In the right Figure \ref{f:22}, we see that the eigenvalue perturbation for $\vert \lambda\vert \le 10^{-2}$ small is linear and for $\vert \lambda \vert \ge 10^{-1}$ quadratic terms are dominant. 
\begin{figure}
\includegraphics[width=5cm]{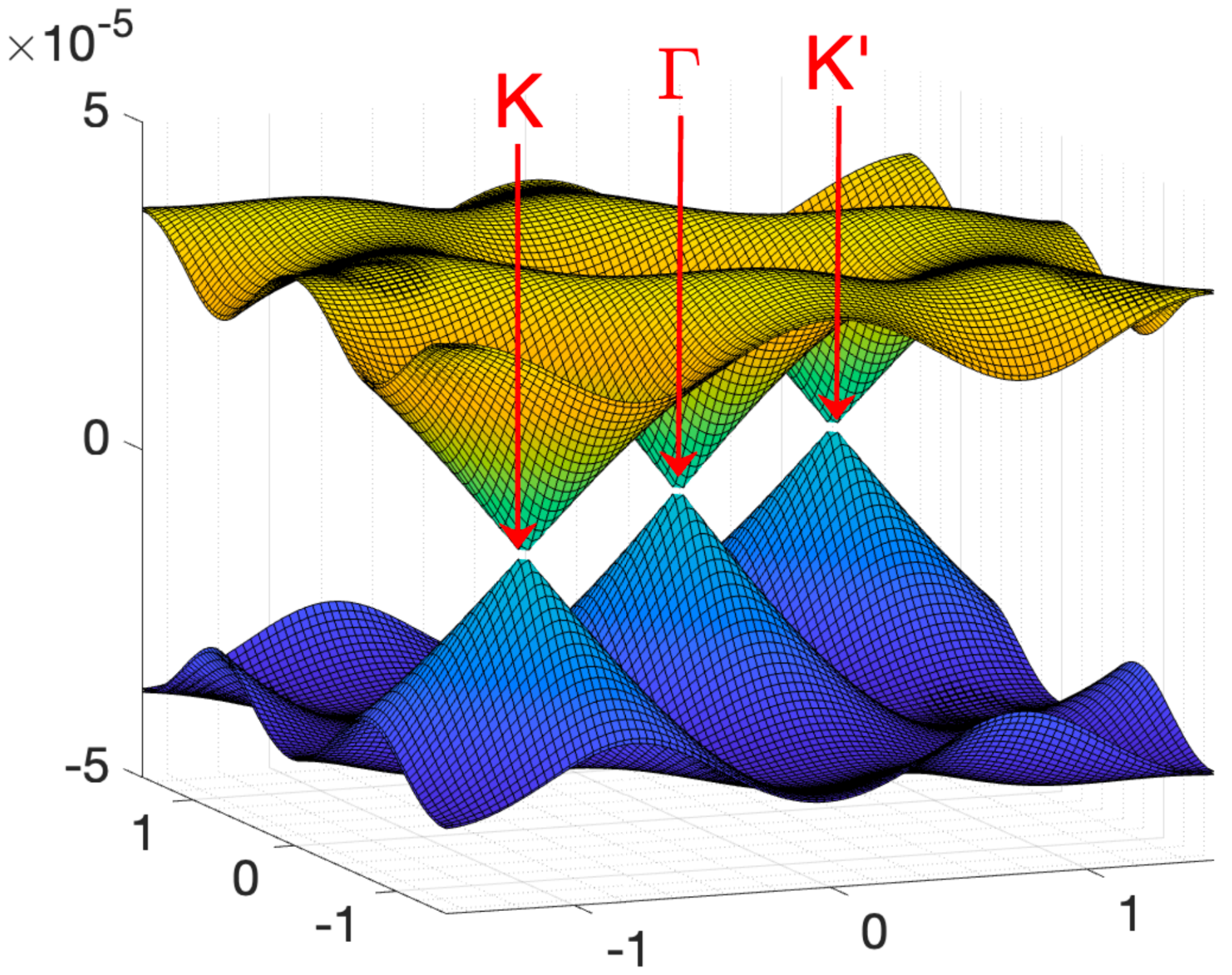}
\includegraphics[width=5cm]{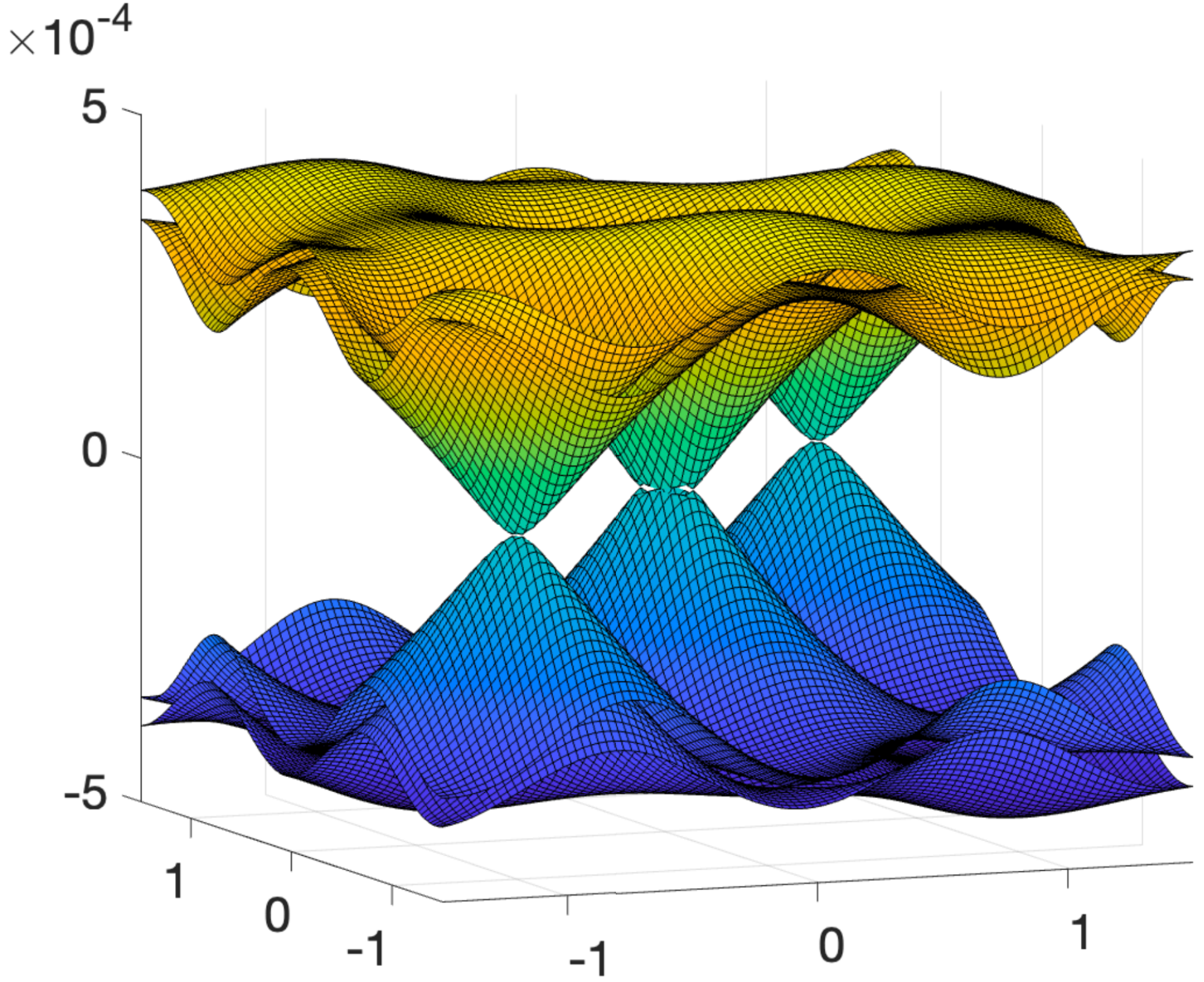}
\includegraphics[width=5cm]{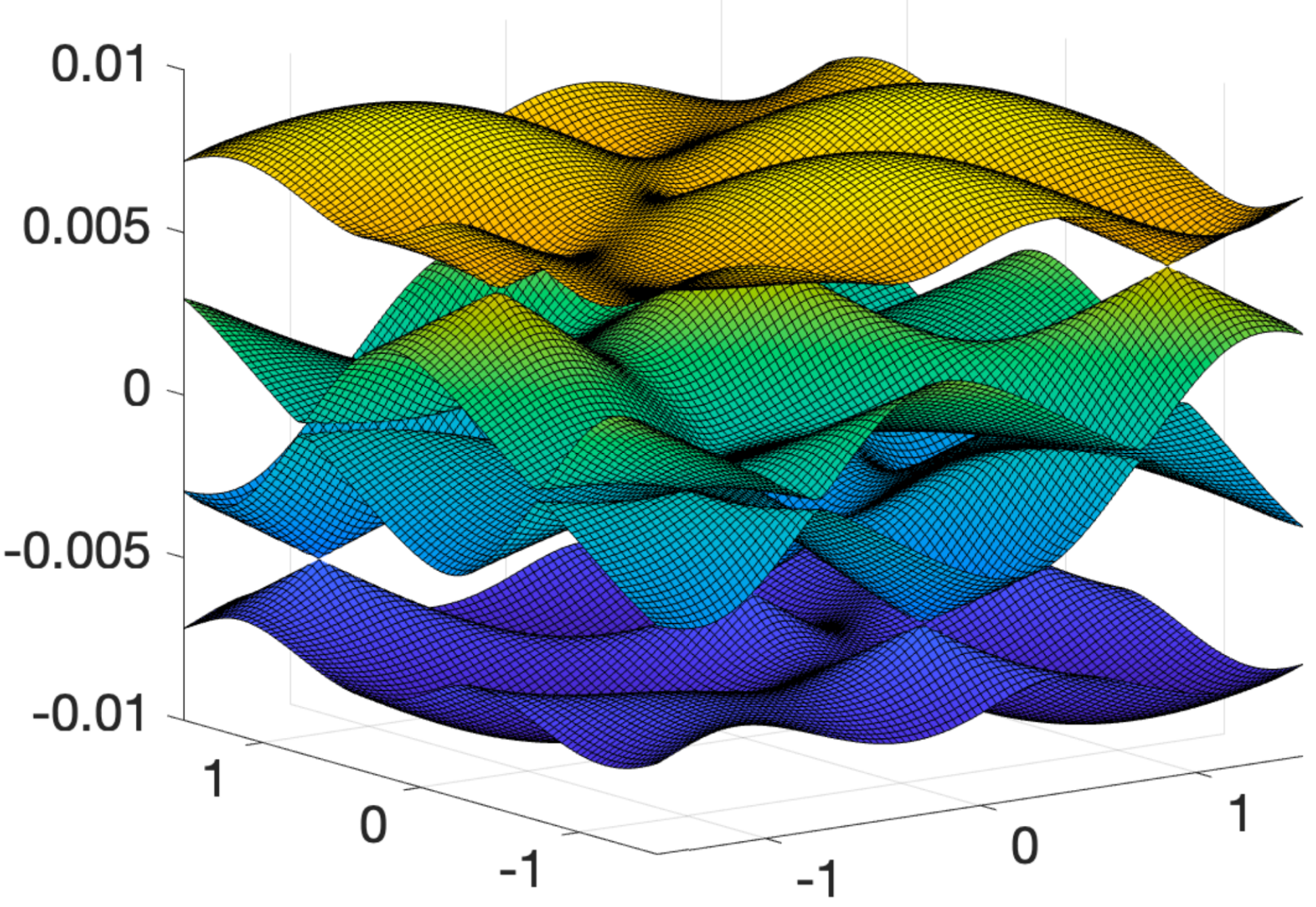}
\caption{\label{f:22} Plots of $ k \mapsto E_{\pm \ell} ( \alpha, \lambda, k ) $ with $\ell \in \{1,2\}$ for 
$ \alpha $ the first real magic element with potential  $AB'$ tunnelling potential $ U_1 ( z ) :=   ( U ( z) - U ( 2z ) )/\sqrt 2 $ and $AA'/BB'$ potential $ V ( z ) $ as in \eqref{eq:defUV}.
In the chiral limit, the potential $U_1$ exhibits double magic $ \alpha$'s on the real axis -- see \cite{bhz3}. We then study 
$ \lambda =  10^{-3}, 10^{-2}, 10^{-1 } $ (from left to right). We see that for very small coupling two of the four bands
"move together" and split only when the coupling gets larger.
For an animated version see 
\url{https://math.berkeley.edu/~zworski/2bands_1D.mp4}
}
\end{figure}

\end{document}